\documentclass[lettersize,journal]{IEEEtran}

\usepackage{amsmath,amsfonts}
\usepackage{amssymb}

\usepackage{subfigure}
\usepackage[linesnumbered, ruled, vlined]{algorithm2e}
\usepackage{textcomp} 
\usepackage{stfloats}
\usepackage{url}
\usepackage{verbatim}
\usepackage{graphicx}
\usepackage{cite}

\usepackage{algpseudocode}
\usepackage{picins}
\usepackage{color}

\usepackage{balance}

\newtheorem{proposition}{Proposition}
\newtheorem{theorem}{Theorem}
\newtheorem{corollary}{Corollary}

\newtheorem{remark}{Remark}
\newtheorem{assump}{Assumption}

\newtheorem{definition}{Definition}

\allowdisplaybreaks

\DeclareMathOperator*{\argmin}{arg\,min}

\def\BibTeX{{\rm B\kern-.05em{\sc i\kern-.025em b}\kern-.08em
    T\kern-.1667em\lower.7ex\hbox{E}\kern-.125emX}}

\begin{document}
\title{Optimal Strategy Revision in Population Games: \\ A Mean Field Game Theory Perspective}
\author{Julian Barreiro-Gomez and Shinkyu Park
\thanks{}
\thanks{The work of Barreiro-Gomez is supported by the Khalifa University Faculty Start-Up (FSU) Grant 8474000774. The work of Park was supported by funding from King Abdullah University of Science and Technology (KAUST). 
} 
\thanks{Barreiro-Gomez is with the Department of Computer \& Information Engineering, Khalifa University, and Khalifa University Center for Autonomous Robotics System (KUCARS), Abu Dhabi, UAE. {\tt julian.barreirogomez@ku.ac.ae}.}
\thanks{Park is with the Department of Electrical and Computer Engineering, King Abdullah University of Science and Technology (KAUST), Thuwal, 23955, Saudi Arabia. {\tt shinkyu.park@kaust.edu.sa}.}
\thanks{Barreiro-Gomez is deeply grateful to God and to the intercession of Our Lady of Lourdes for the blessings of health and life, without which his contributions to this work would not have been possible.}}

\maketitle

\begin{abstract}
This paper investigates the design of optimal strategy revision in Population Games (PG) by establishing its connection to finite-state Mean Field Games (MFG). Specifically, by linking Evolutionary Dynamics (ED)---which models agent decision-making in PG---to the MFG framework, we demonstrate that optimal strategy revision can be derived by solving the forward Fokker-Planck (FP) equation and the backward Hamilton-Jacobi (HJ) equation, both central components of the MFG framework. Furthermore, we show that the resulting optimal strategy revision, which maximizes each agent's payoffs over a finite time horizon, satisfies two key properties: positive correlation and Nash stationarity, which are essential for ensuring convergence to the Nash equilibrium. This convergence is then rigorously analyzed and established. Additionally, we discuss how different design objectives for the optimal strategy revision can recover  existing ED models previously reported in the PG literature. Numerical examples are provided to illustrate the effectiveness and improved convergence properties of the optimal strategy revision design.

\end{abstract}


\section{Introduction}
Large-population game frameworks are widely applicable across engineering domains such as multi-agent task allocation, demand response in smart grids, wireless communication networks, distributed control systems, and large-scale optimization problems \cite[§III]{Park2019Payoff-Dynamic-}, \cite{7823106}, where numerous agents engage in strategic interactions. In particular, evolutionary dynamics (ED) describe decision-making models that enable large populations of agents to adapt and learn strategies that converge to a Nash equilibrium in population games (PG). Furthermore, mean field games (MFG) offer an analytically tractable framework for solving stochastic differential games involving large populations \cite{Lasry2007,Caines}.

Both PG and MFG provide rigorous frameworks for designing decision-making models applicable to scenarios involving finite but large populations of agents. Although they share this overarching objective, they differ fundamentally in their modeling perspectives. The PG framework focuses on deriving decision models for large populations of agents engaged in myopic decision-making based on instantaneous payoffs \cite{Sandholm2010-SANPGA-2}. These models are governed by \textit{strategy revision protocols}, which specify how agents switch between strategies based on the payoffs associated with their current strategy choices. In contrast, the MFG framework aims to compute solutions that optimize long-term cost functions dependent on the evolving population's state in either deterministic or stochastic differential game settings, considered in the population limit \cite{Lacker2016}.

A substantial body of literature on PG has explored a variety of decision-making models and developed techniques for establishing their stability (e.g., \cite{HOFBAUER20091665, HOFBAUER200747, Fox2013Population-Game, 9219202, 10933512,EDM_water_passivity}). Building on these foundations, recent advances have introduced higher-order dynamics models to refine payoff evaluations \cite{Fox2013Population-Game, park2018cdc}. Such models can, for example, be designed to filter noise in perceived payoffs \cite{Park2019Payoff-Dynamic-} or mitigate the adverse effects of time delays on convergence to the Nash equilibrium \cite{Park2023LearningWD}. Other works, such as \cite{9022871}, have established the stability of certain classes of decision-making models in PG by uncovering connections between PG and finite game formulations and leveraging stability analysis techniques developed in the latter context.

Despite these significant advances, a fundamental question remains unresolved: which strategy revision protocol is most effective for achieving equilibrium under a given objective? In this work, we address this question by leveraging the MFG framework to propose a novel approach for designing payoff dynamics models that optimally balance rapid convergence to the Nash equilibrium with the penalties associated with strategy revisions. To the best of our knowledge, this is the first study to formally establish a direct connection between MFG and ED models in PG, thereby enabling the computation of optimal strategy revision protocols.

Solutions for standard MFG problems are characterized by two equations: a backward Hamilton-Jacobi equation, which describes the evolution of the optimal cost-to-go, and a forward Fokker-Planck-Kolmogorov equation, which captures the distribution of agents' states. Traditionally, although both MFG and PG focus on decision-making in large populations, their primary difference lies in the structure of their strategy sets: MFG typically involves a continuous strategy set, as seen in differential games of the mean-field type \cite{CRC_Book_BaTe}, whereas PG is defined over a finite number of strategies \cite{Sandholm2010-SANPGA-2}. 

The study by \cite{Diogo_2013} introduces a finite-state MFG framework, laying the groundwork for establishing a connection between MFG and PG. Building on this foundation, the primary contribution of our work is to formally establish the relationship between MFG and PG frameworks. Specifically, we demonstrate the connection between the Fokker-Planck (FP) equation in finite-state MFG and the ED model in PG. Leveraging this connection, we show that the optimal solution in the finite-state MFG framework corresponds to an optimal strategy revision protocol in PG, with payoff dynamics derived from the backward Hamilton-Jacobi (HJ) equation. 

The contributions of this paper can be summarized as follows: $i)$ We establish a connection between the finite-state MFG introduced in \cite{Diogo_2013} and the ED models presented in \cite{Sandholm2010-SANPGA-2, SANDHOLM2015703}. This is achieved by introducing a dynamic model for payoff vectors analogous to the backward HJ equation in MFG. $ii)$ We demonstrate that, under the proposed payoff dynamics model, changes in the distribution of the agents' strategy selections exhibit positive correlation with the payoff vector. This behavior aligns with the positive correlation property observed in conventional ED models and is critical for ensuring convergence to the Nash equilibrium. Additionally, we establish the link between the steady-state solution of the MFG model and the Nash equilibrium in PG. The convergence of the population state to the Nash equilibrium is rigorously discussed. $iii)$ We show that existing ED models, including Smith, replicator, and projection dynamics, can be derived as special cases of our approach.

This paper is organized as follows: $\S$\ref{sec:brief_background} provides the problem statement and a brief overview of preliminaries on PG and finite-state MFG. $\S$\ref{sec:main} presents the main results of this paper. Illustrative examples are discussed in $\S$\ref{sec:examples}. Finally, conclusions are drawn in $\S$\ref{sec:conclusions}.

\section{Preliminaries and Problem Statement}
\label{sec:brief_background}

\subsection{Population Games and Problem Formulation}

Consider a large population of decision-making agents, each selecting a strategy from the set $\mathbb{S} = \{1,\cdots,n\}$. At any time $t \in [t_0,T]$ with $0 \leq t_0<T$, agents may revise their strategy choices. The vector-valued function $x: [t_0,T] \to \Delta$ represents \textit{the state of the population}, describing the distribution of agents' strategy selections at a given time $t$. The $i$-th entry $x_i(t)$ of $x(t)$ specifies the proportion of agents selecting strategy $i \in \mathbb{S}$ at time $t$. The set $\Delta$, known as the probability simplex, is defined as $\Delta = \{x \in \mathbb R_+^n \,|\, \sum_{i \in \mathbb S} x_i = 1\}$, representing the set of all feasible population states. The relative interior of the simplex is given by $\mathrm{int}(\Delta) = \{x \in \Delta \,|\, x_i > 0, \forall i \in \mathbb S \}$. 

In standard population games, the payoff function $F: \Delta \to \mathbb R^n$ specifies the payoff vector $F(x(t))$ associated with the population state $x(t)$. For a given $F$, the Nash equilibrium is defined as follows.
\begin{definition}
\label{def:NE}
    A population state $x^* \in \Delta$ is called a \textit{Nash equilibrium} of $F$ if the following holds for all $i \in \mathbb S$:
    \begin{align*}
         x_i^* >0 \implies F_i(x^*) = \textstyle\max_{j \in \mathbb S} F_j(x^*).
    \end{align*}
    The set of all Nash equilibria is denoted by $\mathbb{NE}(F)$. \hfill $\square$
\end{definition}
\vspace{0.06cm}

The definition of contractive population games is given as follows.
\begin{definition} \label{def:contractive_games}
    A population game $F$ is called \textit{contractive} if the following inequality:
    \begin{align}
        \left( F(x) - F(y) \right)^\top \left( x - y \right) \leq 0
    \end{align}
    holds for all $x,y \in \Delta$.  \hfill $\square$
\end{definition}

Agents are allowed to revise their strategy selections at the arrival times of i.i.d. Poisson processes \cite{Sandholm2010-SANPGA-2}. Let $\rho_{ji} (p, x)$ denote the \textit{strategy revision protocol}, which specifies the rate at which a strategy-revising agent switches from strategy $j$ to strategy $i$. For instance, the Smith protocol is defined as $\rho_{ji}( p, x) = [p_i - p_j]_+$, where $[\cdot]_+ = \max(0,\cdot)$.
As shown in \cite{Sandholm2010-SANPGA-2}, under the assumptions of a homogeneous population---where all agents adopt the same revision protocol---and an infinitely large population, the population state $x(t)$ evolves according to the following ordinary differential equation, commonly referred to as the Evolutionary Dynamics (ED):
\begin{multline} \label{eq:evolutionary_dynamics}
    \dot{x}_i(t) = \textstyle \sum_{j \in \mathbb S}  x_j(t) \rho_{ji} (p(t), x(t)) \\ - x_i(t)  \textstyle \sum_{j \in \mathbb S}  \rho_{ij} (p(t), x(t)), ~ i \in \mathbb S,
\end{multline}   
where $p(t) = F(x(t))$.
Let $V$ be the vector field defining $\dot x(t) = V (p(t), x(t))$, where $x(t) = (x_1(t), \cdots, x_n(t))$. As discussed in \cite{HOFBAUER20091665}, the following key properties of ED are crucial in establishing convergence to the Nash equilibrium.
\vspace{0.06cm}
\begin{definition} \label{def:two_properties}
 ED \eqref{eq:evolutionary_dynamics} satisfies $i)$ the \textit{Nash stationarity} if $V(F(x), x) = 0$ holds if and only if $x \in \mathbb{NE}(F)$, and 
    $ii)$ the \textit{positive correlation} if $V (p, x) \neq 0$ implies $p^\top V (p, x) > 0$.  \hfill $\square$
\end{definition}
\vspace{0.06cm}

The primary objective of this work is to design an optimal strategy revision protocol $\rho = (\rho_{ij})_{i,j \in \mathbb S}$ and a payoff dynamics model
\begin{align} \label{eq:payoff_dynamics_model}
    \dot p_i(t) = G_i (p(t), x(t)), ~ i \in \mathbb S,
\end{align}
that maximize the agents' long-term payoffs. 
The problem is formulated as an optimal control problem for a \textit{reference} agent initialized with strategy~$i$. The reference agent is defined as a randomly selected individual whose strategy at the initial time~$t_0$ is $i$.
As in the standard framework of population games, each agent revises its strategy based on its own payoff. The effect of an agent’s strategic decisions on its future payoffs is captured through the population state $x(t)$ in the formulation. Consequently, the finite time-horizon payoff function is defined exclusively in terms of the current and anticipated future payoffs of the agent.

\textbf{Main Problem:} Consider an optimal control problem where a reference agent initially selects strategy $i$ at time $t_0$ and subsequently revises its strategy according to the strategy revision protocol. The objective is to design $\rho = (\rho_{ij})_{i,j \in \mathbb S}$ and $G = (G_i)_{i \in \mathbb S}$ to \textit{maximize} a finite time-horizon payoff function defined as:\footnote{In population games, agents are non-atomic, meaning their individual identities are negligible; only their strategy choices matter. Consequently, the finite time-horizon payoff function \eqref{eq:cost_to_go} describes the long-term payoffs received by agents who begin with strategy~$i$ under the revision protocol $\rho$.}
\begin{multline} \label{eq:cost_to_go} 
    P_i(t_0, \rho) = \mathbb{E} \Big[ \textstyle\int_{t_0}^{T} \big( -\sum_{j \in \mathbb {S}} \frac{q_{s(t), j}(x(t))}{2} \rho_{s(t), j}^2 (p(t), x(t)) \\
    + F_{s(t)}(x(t)) \big) \mathrm{d}t + F_{s(T)}(x(T)) \,\Big|\, s(t_0) = i \Big],
\end{multline} 
where $F = (F_1, \cdots, F_n)$ is the instantaneous payoff function of the underlying population game, $s(t) \in \mathbb S$ denotes the agent's strategy at time $t$, and $x(t)$ and $p(t)$ are determined by \eqref{eq:evolutionary_dynamics} and \eqref{eq:payoff_dynamics_model}, respectively. 
We assume that $F$ is continuously differentiable on $\Delta$. The expectation in \eqref{eq:cost_to_go} is taken with respect to the strategy selection $s(t)$, a random variable over $\mathbb S$ with the transition probability satisfying \cite[\S 4]{park2018cdc}: $\lim_{\epsilon \to 0} \frac{1}{\epsilon} \mathbb{P} \left[ s(t + \varepsilon) = j \,\big|\, s(t) = i \right] = \rho_{ij}(p(t), x(t))$.

The finite time-horizon payoff function $P_i(t_0, \rho)$ accumulates the payoffs over the interval $[t_0, T]$, incorporating a penalty term that grows quadratically with the revision protocol $\rho_{ij}$. The weighting function $q_{ij}(x(t))$, defined for all $i,j \in \mathbb{S}$ and depending on the population state $x(t)$, plays a central role in shaping the optimal revision protocol. A smaller value of $q_{ij}(x(t))$ reduces the cost associated with revising the strategy from $i$ to $j$, thereby accelerating the ED model described in \eqref{eq:evolutionary_dynamics}. We assume that $q_{ij}(x) > 0$ for all $x \in \mathrm{int}(\Delta)$, that $q_{ij}(x)$ is continuously differentiable on $\mathrm{int}(\Delta)$, and that $q_{ii}(x) = 0$ for all $x \in \Delta$. Examples of such weighting functions are discussed in $\S$\ref{sec:link_to_existing_ed}.\footnote{The formulation can also be extended to cases where $q_{ij}(x(t))$ becomes arbitrarily large---for instance, when $q_{ij}(x(t)) = 1/a_{ij}$ with $a_{ij} = 0$---as will be discussed in Corollary~\ref{coro:2}. In this case, the analysis proceeds by first setting $\rho_{ij}(p(t), x(t)) = 0$, and then determining the revision protocols $\rho_{kl}(p(t), x(t))$ for which $q_{kl}(x(t))$ remains bounded.}

\subsection{Finite-State Mean Field Games}
\label{sec:mean-field}

Consider a large population of homogeneous agents interacting within a finite-state MFG framework over a time horizon $[t_0,T]$ with $0 \leq t_0<T$, as discussed in \cite{Diogo_2013}. In this framework, the state of a \textit{reference} agent is modeled by a continuous-time Markov chain defined over a finite state space $\mathbb{S}=\{1,\cdots,n\}$. State transitions are governed by an $n \times n$ transition matrix $\alpha (t) = (\alpha_{ij}(t))_{i,j \in \mathbb S} \in \mathbb R^{n \times n}$, satisfying the following conditions at each time~$t$: $\alpha_{ij}(t) \geq 0, ~ \forall i,j \in \mathbb S$, $i \neq j$, and $\sum_{j \in \mathbb{S}} \alpha_{ij}(t) = 0, ~ \forall i \in \mathbb{S}$. More formally, for $\varepsilon > 0$, the transition probability is defined as
\begin{align}
\label{eq:Markov}
    \mathbb{P} \left[ s(t + \varepsilon) = j \,\big|\, s(t) = i \right] = \alpha_{ij}(t) \varepsilon + o(\varepsilon),
\end{align}
where $s(t) \in \mathbb S$ represents the agent's state at time~$t$, and $\lim_{\varepsilon \to 0} o(\varepsilon)/\varepsilon = 0$. 

The state distribution across $\mathbb S$ is represented by a function $x: [t_0,T] \to \Delta$, referred to as the \textit{mean-field term}. The main objective in MFG problems is to compute the transition matrix $\alpha (t), \, t \in [t_0, T]$ that \textit{minimizes} a finite time-horizon cost function $J_i$ for a reference agent starting in state $i$ at time $t_0$, defined as:
\begin{multline} \label{eq:cost_mfg}
J_i(t_0, \alpha) = \mathbb{E} \Big[ \textstyle\int_{t_0}^{T} c(s(t),x(t), \alpha_{s(t)} (t)) \, \mathrm{d}t  \\
+ \psi(s(T),x(T)) \,\Big|\, s(t_0) = i \Big], 
\end{multline}
where $\alpha_i(t) = (\alpha_{i1}(t), \cdots, \alpha_{in}(t))$ is the $i$-th row of the transition matrix, $c: \mathbb S \times \Delta \times \mathbb R^n \to \mathbb R$ is the running cost, and $\psi: \mathbb S \times \Delta \to \mathbb R$ is the terminal cost. For notational convenience, given $i \in \mathbb S$, let $\mathbb R_{i,0}^n$ denote the set of $n$-dimensional vectors $\mu = (\mu_1, \cdots, \mu_n)$ satisfying: $i)$ $\mu_j \geq 0$ if $j \neq i$, and $ii)$ $\sum_{j \in \mathbb S} \mu_j = 0$. The optimal cost-to-go from time $t$ is defined as $w_i(t,x) = \min_{\alpha} J_i(t,\alpha)$.
\begin{assump} \label{assump:mfg}
It is assumed that $1)$ the running cost $c(i, x, \mu)$ is Lipschitz continuous in $x \in \Delta$, with the Lipschitz constant (with respect to $x$) independent of $\mu \in \mathbb R_{i,0}^n$, $2)$ $c(i, x, \mu)$ is differentiable with respect to $\mu \in \mathbb R_{i,0}^n$, $3)$ for any $i \in \mathbb S$, $x \in \Delta$, the partial derivative $\frac{\partial c}{\partial \mu}(i, x, \mu)$ is Lipschitz continuous with respect to $x \in \Delta$, uniformly in $\mu$, $4)$ $c(i, x, \mu)$ does not depend on $\mu_i$, and is uniformly convex: For any $i \in \mathbb S$, $x \in \Delta$, $\mu, \mu' \in \mathbb R_{i,0}^n$ with $\mu_j \neq \mu_j'$ for some $j \neq i$, it holds that
\begin{align*}
  c(i, x,\mu') \geq c(i,x,\mu) + \nabla_{\mu}^\top c(i,x,\mu) (\mu' \!-\! \mu) + \gamma \| \mu' \!-\! \mu \|^2  
\end{align*}
for some positive constant $\gamma$, $5)$ $c(i, x, \mu)$ is  superlinear on $\mu_j, \, j \neq i$: $\lim_{\mu_j \to \infty} \frac{c(i, x, \mu)}{\|\mu\|} = \infty$,
and $6)$ the terminal cost $\psi(i, x)$ is Lipschitz continuous in $x$. \hfill $\square$
\end{assump}

The solution to the MFG problem is given by two coupled equations, as outlined in \cite{Diogo_2013}. The first is the backward evolution of $w$, known as the \textit{Hamilton-Jacobi (HJ) equation}:
\begin{subequations}
\label{eq:mean_field system}
\begin{align}
&\dot{w}_i(t,x(t)) = -H_i(w(t, x(t)),x(t)) \label{eq:mean_field system_a}\\
&w_i(T,x(T)) =  \psi(i,x(T)),\label{eq:mean_field system_b}
\end{align}
\end{subequations}
where the Hamiltonian $H_i$ is defined as
\begin{align*}
    H_i(z,x) = \textstyle \min_{\mu \in \mathbb R_{i,0}^n} \left( c(i,x, \mu) + \sum_{j \in \mathbb{S}} \mu_j (z_j - z_i) \right).
\end{align*}
The second equation describes the forward evolution of $x(t)$, commonly referred to as the \textit{Fokker-Planck (FP) equation}:
\begin{align}
\begin{array}{l}
\dot{x}_i(t) = \textstyle\sum_{j \in \mathbb{S}} \alpha_{ji} (t) x_{j}(t), ~ x_i(t_0) = x_{0,i},      
\end{array} 
\label{eq:mean_field system_c}
\end{align}
where $x_0 = (x_{0,1}, \cdots, x_{0,n}) \in \Delta$, and 
\begin{multline*}
    \alpha_i (t) = \textstyle \argmin_{\mu \in \mathbb R_{i,0}^n} \big ( c(i, x(t), \mu) \\ + \textstyle \sum_{j \in \mathbb S } \mu_j (w_j(t, x(t)) - w_i (t, x(t)) ) \big),
\end{multline*}
representing the optimal state transition. 
The existence of solutions to \eqref{eq:mean_field system} and \eqref{eq:mean_field system_c} is established in \cite[Proposition 4]{Diogo_2013}, while their uniqueness is proven in \cite[Theorem 2]{Diogo_2013}.

The MFG system \eqref{eq:mean_field system} and \eqref{eq:mean_field system_c} may not admit a conventional stationary solution, i.e., one satisfying $\dot{w}(t, x(t)) = \dot{x}(t) = 0$, for instance, in cases when $H_i > 0$. Instead, following the approach in \cite{Diogo_2013}, we define a stationary solution \textit{modulo addition of a constant $\kappa$}, characterized by $\dot{w}(t, x(t)) = -\kappa \mathbf 1_n$ and $\dot{x}(t) = 0$, for which the following conditions hold for all $i \in \mathbb S$:
\begin{subequations}
\label{eq:mean_field_system_stationary}
\begin{align}
H_i(w(t, x(t)), x(t)) &= \kappa \label{eq:mean_field system_stationary_a}\\
\textstyle\sum_{j \in \mathbb{S}} \alpha_{ji}(t) x_{j}(t) &= 0. \label{eq:mean_field system_stationary_b}
\end{align}
\end{subequations}    
Note that the existence of such a stationary solution follows from \cite[Proposition~5]{Diogo_2013}.

\section{Optimal Strategy Revision Protocol Design}
\label{sec:main}

\begin{figure}[t!]
    \centering
    \subfigure[]{
    \includegraphics[trim={0.0in 0.0in 0.0in 0.0in}, clip ,width=1.6in]{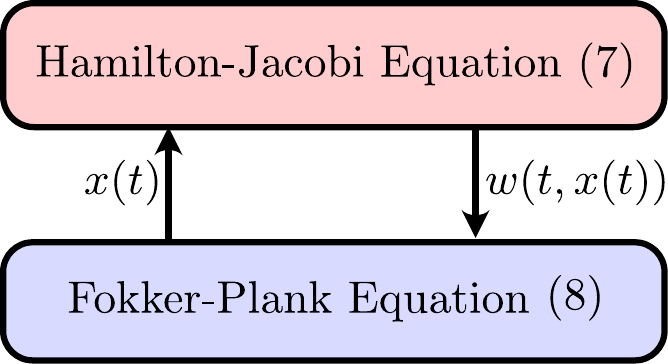}}
    \subfigure[]{
    \includegraphics[trim={0.0in 0.0in 0.0in 0.0in}, clip ,width=1.6in]{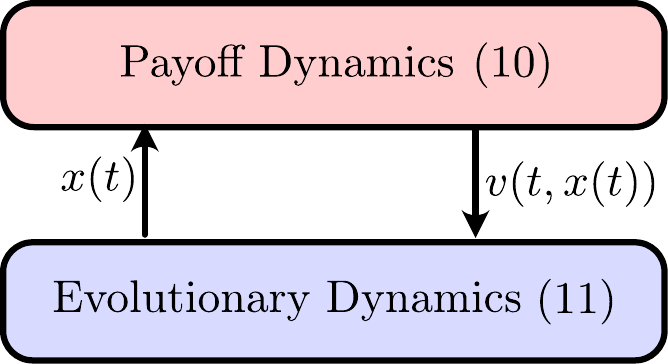}}
    \caption{Closed-loop model diagrams for (a) MFG and (b) the payoff dynamics for the optimal strategy revision in PG. To highlight the relationship, recall that $v_i(t, x(t)) = -w_i(t, x(t))$.}
    \label{fig:schemes}
\end{figure}

We begin by establishing the relationship between the MFG system and the optimal strategy revision process in PG, as illustrated in Fig.~\ref{fig:schemes}. The following theorem establishes the connection between the state distribution dynamics in finite-state MFGs and the ED model, providing a solution to our main problem. The proof is presented in Appendix~\ref{proof:propos:connection}.
We note that the existence and uniqueness of solutions to the main problem are ensured by the results established for \eqref{eq:cost_mfg}, as discussed in \S\ref{sec:mean-field}.

\begin{theorem}
\label{propos:connection}
The optimal strategy revision protocol for \eqref{eq:cost_to_go} is given by:
\begin{align*}
    \rho_{ij}(p(t), x(t)) = \frac{[p_j(t) - p_i(t)]_+}{q_{ij}(x(t))},
\end{align*}
where $p_i(t) = v_i(t, x(t))$, and $v_i(t, x(t))$ satisfies the following backward differential equation:
\begin{subequations}
\label{eq:HJ_equation}
\begin{align}
\dot{v}_i(t,x(t)) &= -\frac{1}{2} \sum_{j \in \mathbb{S}} \dfrac{[v_j(t,x(t)) \!-\! v_i(t,x(t))]_+^2}{q_{ij}(x(t))} \!-\! F_i(x(t)) \\
v_i(T,x(T)) &= F_i(x(T)).
\end{align}
\end{subequations}
The corresponding evolution of $x_i(t)$ is described by:
\begin{subequations}
\label{eq:pairwise}
\begin{align}
\dot{x}_i(t) &= \textstyle\sum_{j \in \mathbb{S}}  x_{j}(t)  \dfrac{[v_i(t,x(t)) - v_j(t,x(t))]_+}{q_{ji}(x(t))} \notag\\
&\qquad - x_{i}(t) \textstyle\sum_{j \in \mathbb{S}}  \dfrac{[v_j(t,x(t)) - v_i(t,x(t))]_+}{q_{ij}(x(t))}, \label{eq:pairwise_a}\\
x_i(t_0) &= x_{0,i}, \label{eq:pairwise_b}
\end{align}
\end{subequations}
where $x_0 \in \Delta$. \hfill $\square$
\end{theorem}

As stated in the theorem, the payoff dynamics model \eqref{eq:payoff_dynamics_model} is given by:
\begin{align*}
    G_i(p(t), x(t)) = -\frac{1}{2} \textstyle \sum_{j \in \mathbb{S}} \dfrac{[p_j(t) - p_i(t)]_+^2}{q_{ij}(x(t))} - F_i(x(t))
\end{align*}
with the terminal condition $p_i(T) = F_i(x(T))$ for all $i \in \mathbb S$. 

\vspace{0.06cm}

\subsection{Direct Connection with Conventional ED Models} \label{sec:link_to_existing_ed}

The following corollary demonstrates how specific classes of existing ED models can be derived from Theorem~\ref{propos:connection} by identifying a corresponding weight $q_{ij}(x(t))$ in \eqref{eq:cost_to_go} for each class. Each resulting ED model is named after its counterpart in the literature, as discussed in Remark~\ref{rmk:1}, where we provide an interpretation of the associated revision protocol. This interpretation offers valuable insights into the behaviors observed in PG. For brevity, we restrict attention to the case where $x(t) \in \mathrm{int}(\Delta)$ for all $t \in [t_0, T]$.

\vspace{0.06cm}

\begin{corollary}
\label{lemma:link2}
    Consider the following choices for the weight $q_{ij}(x(t))$ in \eqref{eq:cost_to_go}: $i)$ $q_{ij}(x(t)) = 1$, $ii)$ $q_{ij}(x(t)) = 1/x_j(t)$, and $iii)$ $q_{ij}(x(t)) = x_i(t)$. For each case, the corresponding ED models are as follows:
    
    $i)$ \textit{Smith dynamics and payoff dynamics ($q_{ij}(x(t)) = 1$):}
    \begin{align*}
        \dot{x}_i(t) &= \textstyle\sum_{j \in \mathbb{S}} x_{j}(t)[v_i(t,x(t)) - v_j(t,x(t))]_+ \\
        &\qquad - x_{i}(t)  \textstyle\sum_{j \in \mathbb{S}} [v_j(t,x(t)) - v_i(t,x(t))]_+ \\
        \dot{v}_i(t,\!x(t)) &\!=\! -\frac{1}{2} \textstyle\sum_{j \in \mathbb{S}} [v_j(t,\!x(t)) \!-\! v_i(t,\!x(t))]_+^2 \!-\! F_i(x(t)).
    \end{align*}
    
    $ii)$ \textit{Replicator dynamics and payoff dynamics ($q_{ij}(x(t)) = 1/x_j(t)$):}
    \begin{align*}
        \dot{x}_i(t) &= x_i(t) \big( v_i(t, x(t)) - \textstyle\sum_{j \in \mathbb{S}} x_{j}(t) v_j(t, x(t)) \big) \\
        \dot{v}_i(t, x(t)) &= -\frac{1}{2} \textstyle\sum_{j \in \mathbb{S}} x_j(t)[v_j(t,x(t)) - v_i(t,x(t))]_+^2  \\
        &\qquad - F_i(x(t)).
    \end{align*}
    
    $iii)$ \textit{Projection dynamics and payoff dynamics ($q_{ij}(x(t)) = x_i(t)$):}
    \begin{align*}
        \dot{x}_i(t) &= \textstyle\sum_{j \in \mathbb{S}}  \big( v_i(t,x(t)) - v_j(t,x(t)) \big) \\
        \dot{v}_i(t,\!x(t)) &\!=\! -\frac{1}{2} \textstyle\sum_{j \in \mathbb{S}} \dfrac{[v_j(t,\!x(t)) \!-\! v_i(t, x(t))]_+^2}{x_i(t)} - F_i(x(t)).
    \end{align*}
    All models satisfy the terminal condition $v_i(T,x(T)) =  F_i(x(T))$ and the initial condition $x(t_0) = x_0 \in \mathrm{int}(\Delta)$. \hfill $\square$
\end{corollary}

\vspace{0.06cm}

Corollary \ref{lemma:link2} highlights interesting connections between the choice of $q_{ij}(x(t))$ and several well-known properties in ED, as elaborated in the following remarks.

\vspace{0.06cm}

\begin{remark}
\label{rmk:1}
Note the following two facts:

$i)$ For the replicator dynamics, when $\rho_{ij}(p(t), x(t)) > 0$, \eqref{eq:cost_to_go} may diverge as the proportion of decision-makers approaches zero, i.e., $q_{ij}(x(t)) \to \infty$ when $x_j(t) \to 0$. Consequently, choosing $q_{ij}(x(t))={1}/{x_j(t)}$ serves as a penalty function that ensures the boundary $\partial \Delta$ of the simplex remains invariant. This behavior is consistent with the well-known property that, under replicator dynamics, the boundary $\partial \Delta$ forms an invariant set \cite{Hofbauer_1987}.

$ii)$ As $T \to t_0$ in \eqref{eq:HJ_equation}, we recover $v_i(t_0, x(t_0)) = F_i(x(t_0))$. Substituting this result into \eqref{eq:pairwise} reproduces the standard ED models, as in \eqref{eq:evolutionary_dynamics}, widely studied in the literature \cite{Sandholm2010-SANPGA-2,BaObQu2017,EDM_constraints}. In contrast, when $T > t_0$, the presence of \textit{payoff dynamics}---rather than static payoff functions---allows agents to optimize based on their cost-to-go, thereby exhibiting forward-looking strategic decision-making. \hfill $\square$
\end{remark}

In the following corollary, we consider migration constraints that restrict agents from switching between specific pairs of strategies. These permissible strategy switches are represented by an undirected connected graph $(\mathbb{S},\mathbb{L},A)$, where $\mathbb{L} \subseteq \{ (i,j) \,|\, i,j \in \mathbb{S} \}$ denotes the set of allowed migrations (graph links), and a symmetric matrix $A \in \{0, 1\}^{n \times n}$ is the adjacency matrix. The entries $a_{ij}$ of $A$ capture the migration constraints, such that $(i,j) \in \mathbb{L}$ if and only if $a_{ij} = 1$. The set of neighbors of strategy $i$ is defined as $\mathbb{S}_i = \{j \in \mathbb S \,|\, (j,i) \in \mathbb{L} \}$. 
Prior works \cite{BaObQu2017, EDM_constraints, Tembine_2008} have demonstrated that migration constraints lead to distributed population dynamics by incorporating the adjacency matrix $A$ -- or equivalently, the migration constraints -- into the revision protocol.
\begin{corollary}
\label{coro:2}
Consider the following choices for the weights $q_{ij}(x(t))$ in \eqref{eq:cost_to_go} using $A = (a_{ij})_{i,j \in \mathbb S}$: $i)$ $q_{ij}(x(t))={1}/{a_{ij}}$, $ii)$ $q_{ij}(x(t)) = {1}/{(a_{ij} x_j(t))}$, and $iii)$ $q_{ij}(x(t)) = {x_i(t)}/{a_{ij}}$. For each case, the corresponding ED models are as follows:

$i)$ \textit{Distributed Smith dynamics and payoff dynamics ($q_{ij}(x(t)) = 1/a_{ij}$):}
    \begin{align*}
        \dot{x}_i(t) &= \textstyle \sum_{j \in \mathbb{S}_i}   x_{j}(t)[v_i(t,x(t)) - v_j(t,x(t))]_+ \\
        &\qquad - x_{i}(t)  \textstyle \sum_{j \in \mathbb{S}_i} [v_j(t,x(t)) \!-\! v_i(t,x(t))]_+ \\
        \dot{v}_i(t,x(t)) &= -\frac{1}{2} \textstyle \sum_{j \in \mathbb{S}_i} [v_j(t,x(t)) - v_i(t,x(t))]_+^2 \!-\! F_i(x(t)).
    \end{align*}

$ii)$ \textit{Distributed replicator dynamics and payoff dynamics ($q_{ij}(x(t)) = 1/(a_{ij}x_j(t))$):}
    \begin{align*}
        \dot{x}_i(t) &= x_i(t) \big( v_i(t,x(t)) \textstyle \sum_{j \in \mathbb{S}_i}  x_{j}(t) \\
        &\qquad\qquad -  \textstyle \sum_{j \in \mathbb{S}_i}  x_{j}(t) v_j(t,x(t)) \big)\\
        \dot{v}_i(t,x(t)) &= -\frac{1}{2} \textstyle \sum_{j \in \mathbb{S}_i} x_j(t)[v_j(t,x(t)) - v_i(t,x(t))]_+^2  \\
        &\qquad - F_i(x(t)).
    \end{align*}

$iii)$ \textit{Distributed projection dynamics and payoff dynamics ($q_{ij}(x(t)) = x_i(t)/a_{ij}$):}
\begin{align*}
    \dot{x}_i(t) &= \textstyle \sum_{j \in \mathbb{S}_i}  (v_i(t,x(t)) - v_j(t,x(t))) \\
    \dot{v}_i(t, x(t)) &\!=\! -\frac{1}{2}\textstyle \sum_{j \in \mathbb{S}_i} \dfrac{[v_j(t,x(t)) \!-\! v_i(t,x(t))]_+^2}{x_i(t)}  \!-\! F_i(x(t)).
\end{align*}
All models satisfy the terminal condition $v_i(T,x(T)) =  F_i(x(T))$ and the initial condition $x(t_0) = x_0 \in \mathrm{int}(\Delta)$.  \hfill $\square$
\end{corollary}

\subsection{Positive Correlation, Nash Stationarity, and Convergence}
In the next two propositions, we show that the optimal strategy revision protocol satisfies positive correlation and Nash stationarity, similar to standard ED models (see Definition~\ref{def:two_properties})

\vspace{0.06cm}

\begin{proposition}
\label{lemma:PC}
The payoff dynamics model \eqref{eq:HJ_equation} and the ED model \eqref{eq:pairwise} satisfy the following positive correlation: $V (p(t),x(t)) \neq 0$ implies $p^\top(t) V(p(t),x(t)) > 0$, where $p(t) = v(t,x(t))$, for all $t \in [t_0, T]$.  \hfill $\square$
\end{proposition}

\vspace{0.06cm}

\begin{proof}
Using \eqref{eq:pairwise}, we can evaluate $p^\top(t) V(p(t),x(t))$ as follows:
\begin{align*}
&p^\top(t) V(p(t),x(t)) \\
&= \textstyle\sum_{i \in \mathbb{S}} v_i(t,x(t)) \dot{x}_i(t) \\ 
&= \textstyle\sum_{i \in \mathbb{S}}\sum_{j \in \mathbb{S}}  x_{j}(t)  \dfrac{[v_i(t,x(t)) - v_j(t,x(t))]_+^2}{q_{ji}(x(t))}.
\end{align*}
Therefore, we conclude that $p^\top(t) V(p(t),x(t)) \geq 0$, with equality holds only if $V(p(t), x(t)) = 0$. 
This concludes the proof. \hfill $\blacksquare$
\end{proof}

Proposition \ref{lemma:PC} implies that, under the optimal strategy revision protocol $\rho$, the population state evolves to maximize the finite time-horizon payoff function \eqref{eq:cost_to_go}.

\begin{proposition}
\label{lemma:steadystate}
    Consider the modulo-addition stationary solution of the associated MFG system, as explained in \eqref{eq:mean_field_system_stationary}. The stationary solution for \eqref{eq:HJ_equation} and \eqref{eq:pairwise} is given by: $x(t) = x^\ast$ and 
        $v(t, x^\ast) = \kappa  (t-T) \mathbf{1}_n + F(x^\ast), ~ \forall t \in [t_0, T]$, where $x^\ast$ is a Nash equilibrium of $F$.  \hfill $\square$
\end{proposition}

\begin{proof}
    From \eqref{eq:mean_field system_stationary_a}, we have:
\begin{align*}
    \dot{v}(t, x^\ast) = \kappa \mathbf{1}_n \implies v(t, x^\ast) = \kappa (t - T) \mathbf{1}_n  + v(T, x^\ast),
\end{align*}
with $v(T, x^\ast) = F(x^\ast)$. From \eqref{eq:mean_field system_stationary_b}, it follows that $V(p(t), x^\ast) = 0$. 
By Definition~\ref{def:two_properties}, this condition implies $x_i^\ast [F_j(x^\ast) - F_i(x^\ast)]_+ = 0, ~ \forall i,j \in \mathbb S$, thereby establishing that $x^\ast$ is a Nash equilibrium. ~\hfill $\blacksquare$
\end{proof}

In the following theorem, we provide a convergence result for a certain class of contractive population games. The proof is provided in Appendix~\ref{proof:cor:convergence}
\begin{theorem} \label{cor:convergence}
Suppose $q_{ij}(x), \, \forall i,j \in \mathbb S$ is continuously differentiable on $\Delta$, and a population game $F$ is contractive and admits an interior Nash equilibrium. Further, assume that $F$ satisfies
\begin{align} \label{eq:strong_contractiveness}
    \left( F(x) - F(y) \right)^\top (x-y) \leq -\epsilon \| x - y \|^2,
\end{align}
for a positive constant $\epsilon$. 
If $x(T/2)$ denotes the population state at time $t=T/2$, then as $T \to \infty$, we have $\|x(T/2) - x^\ast \| \to 0$, where $x^\ast$ is the Nash equilibrium of $F$.  \hfill $\square$
\end{theorem}

\begin{remark} \label{remark:payoff_function_modification}
    For any contractive population game $F$, modifying $F$ by defining $F_\epsilon(x) = F(x) - \epsilon \, (\ln (x_1+\delta) \, \cdots \, \ln (x_n+\delta) )^\top$ with $\epsilon>0$, where $\delta>0$ is a sufficiently small constant. This modification ensures that $F_\epsilon$ satisfies \eqref{eq:strong_contractiveness} and has an interior Nash equilibrium. Let $x^\ast (\epsilon)$ denote the Nash equilibrium of $F_\epsilon$, which also serves as the $\epsilon$-Nash equilibrium of $F$. Since $F$ is a continuous function of $x$, it follows that:
    \begin{align*}
        \lim_{\epsilon \to 0} x^\ast (\epsilon) = x^\ast,
    \end{align*}
    where $x^\ast$ is the Nash equilibrium of $F$.  \hfill $\square$
\end{remark}

\section{Illustrative Examples}
\label{sec:examples}

\begin{algorithm*}
\DontPrintSemicolon
\caption{Compute solution trajectories $(x^*_{[t_0,T]},v^*_{[t_0,T]})$ of \eqref{eq:HJ_equation} and \eqref{eq:pairwise}.} \label{alg:algorithm}
\KwData{$F$, $x_0 \in \Delta$, $t_0, T \in \mathbb{R}_+$ with $T>t_0$, $N \in \mathbb{Z}_+$, $\varepsilon_f \in \mathbb{R}_+$} 
\KwResult{$x^*_{[t_0,T]},v^*_{[t_0,T]}$}

$k \gets 0$, $e \gets \infty$ \;

$x^{(0)}_{[t_0,T]} \gets $ a solution of $\dot{x}(t) = V(p(t),x(t))$ with $x(t_0)=x_0$ and $p(t) = F(x(t))$\;

\While{$(k < N)$ and $(e>\varepsilon_f)$}  {

$v^{(k+1)}_{[t_0,T]} \gets $ a solution of the HJ equation \eqref{eq:HJ_equation} with a fixed population state trajectory $x_{[t_0,T]}$ as  $x_{[t_0,T]} = x^{(k)}_{[t_0,T]}$ \;
$x'_{[t_0,T]} \gets $ a solution of the FK equation \eqref{eq:pairwise} with a fixed payoff vector trajectory $p_{[t_0, T]}$ as $p_{[t_0, T]} = v^{(k+1)}_{[t_0,T]}$ \;
$x^{(k+1)}_{[t_0,T]} \gets (1-a)x^{(k)}_{[t_0,T]} + a x'_{[t_0,T]}$ \;
$e \gets \int_{t_0}^T \|x^{(k)}(t') - x^{(k+1)}(t')\|_2^2 \, \mathrm{d}t'$ \Comment{To check if the fixed-point is reached} \;
$k \gets k + 1$ \;
}
$(x^*_{[t_0,T]},v^*_{[t_0,T]}) \gets (x^{(k+1)}_{[t_0,T]},v^{(k+1)}_{[t_0,T]})$ 
\end{algorithm*}

We provide numerical examples to illustrate the main results of this paper. To compute the solutions for \eqref{eq:HJ_equation} and \eqref{eq:pairwise}, we use Algorithm~\ref{alg:algorithm}, adapted from the literature \cite{Lauriere2023}. The algorithm is specifically designed to find the fixed point of \eqref{eq:HJ_equation} and \eqref{eq:pairwise}.  

The computation begins by initializing the population state trajectory using a conventional game model, $p = F(x)$. The algorithm then alternates between updating the payoff vector trajectory based on the current population state trajectory and vice versa. To stabilize this alternating process, we employ an exponential moving average strategy (Line~6) when updating the population state trajectory. This approach helps mitigate rapid changes and ensures smoother convergence. Each iteration in the algorithm is denoted by the superscript $(k)$.

We note that Algorithm~1 requires solving two ordinary differential equations (Lines~4–5), whose computational cost depends on the length of the time horizon $T$ and the number of strategies $n$. In addition, decreasing the moving average weight $a$ (Line~6) can slow down convergence, thereby increasing the computational time, although it can also improve the accuracy and stability of the iterative process in the algorithm.

We consider the following two population games:
    \paragraph*{\textit{Congestion game}} The payoff function $F$ is defined as:
\begin{align} \label{eq:congestion_game}
    F(x)= \begin{pmatrix} F_1 (x)\\ 
    F_2 (x) \\ F_3 (x) \\ F_4(x) \\ F_5 (x) \\ F_6 (x) \end{pmatrix} = 
    - \begin{pmatrix} 
    2.5 x_1 + x_2 \\ 
    x_1 + 2.5 x_2 + x_3 + 0.5 x_5 \\ 
    x_2 + 2.5 x_3 \\
    2.5 x_4 + x_5 \\ 
    0.5 x_2 + x_4 + 2.5 x_5 + x_6 \\ 
    x_5 + 2.5 x_6
    \end{pmatrix}.
\end{align}
For \eqref{eq:cost_to_go}, we set
$q_{ij}(x(t)) =1, \forall i,j \in \mathbb{S}$, $t_0=0$, and $T=6$. We apply Algorithm~\ref{alg:algorithm} with $a = 0.01$, $N=100$, and $\varepsilon_f=0$.

    \paragraph*{\textit{Rock-Paper-Scissors (RPS) game}} The payoff function $F$ is defined as:
\begin{align} \label{eq:rps_game}
    F(x)
    = \begin{pmatrix} F_1 (x) \\ 
    F_2 (x) \\ F_3 (x) \end{pmatrix} 
    = \begin{pmatrix} - x_2 + x_3 \\ 
    x_1 - x_3 \\ - x_1 + x_2 \end{pmatrix}.
\end{align}
For \eqref{eq:cost_to_go}, we set
$q_{ij}(x(t)) =1, \forall i,j \in \mathbb{S}$, $t_0=0$, and $T=6$. We apply Algorithm~\ref{alg:algorithm} with $a = 0.001$, $N=6000$, and $\varepsilon_f=0$.

\begin{figure}[t!]
    \centering  
    \subfigure[]{
    \includegraphics[width=0.33\textwidth, trim=0cm 0cm 0cm 0cm]{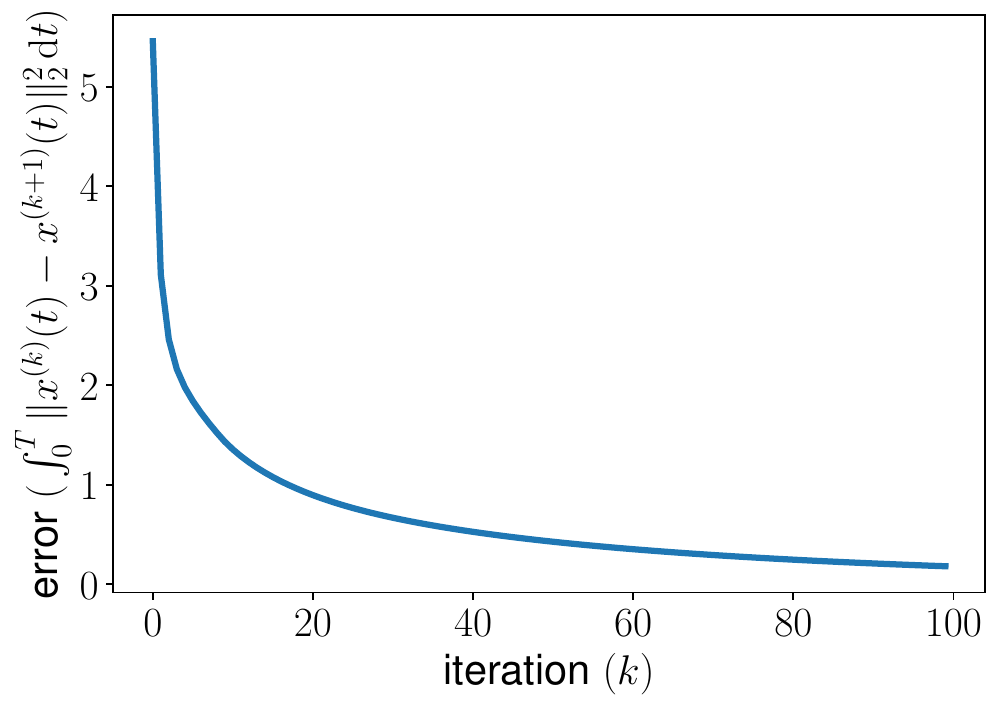} }
    \subfigure[]{\includegraphics[trim={2cm 2cm 0cm 2.5cm}, clip, width=0.32\textwidth]{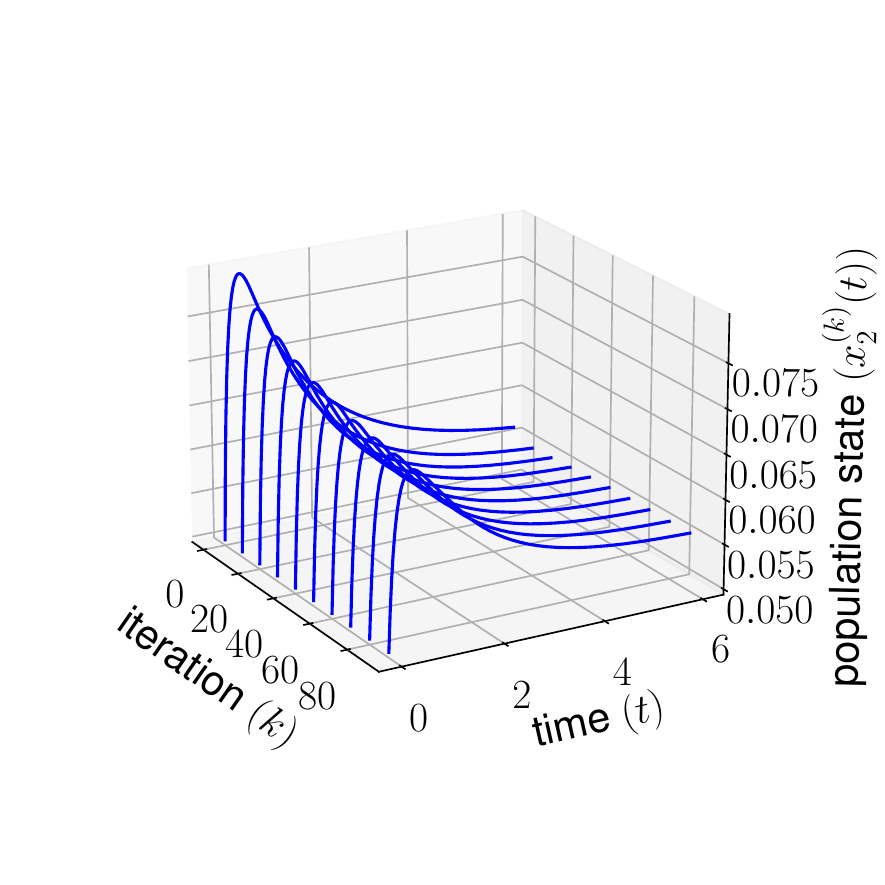}
    \label{fig:congestion_fig_b}}
    \caption{(a) The progression of the error term in Algorithm~\ref{alg:algorithm}, and (b) the resulting trajectory of the population state $x_2^{(k)}(t)$ over various iterations in the congestion game \eqref{eq:congestion_game}.}
    \label{fig:congestion_fig}
\end{figure}

\begin{figure}
    \centering    \includegraphics[width=0.75\linewidth]{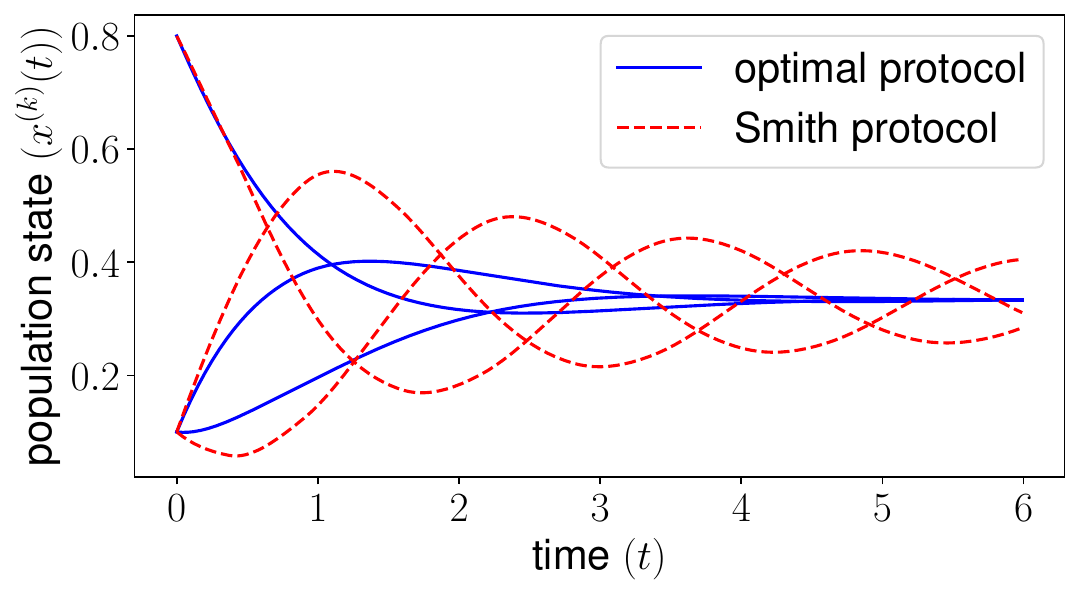}
    \caption{A comparison between population state trajectories determined by the optimal strategy revision protocol and the Smith protocol.}
    \label{fig:RPS_states}
\end{figure}
Fig.~\ref{fig:congestion_fig} illustrates the progression of the error term $\int_{t_0}^T \|x^{(k)}(t') - x^{(k+1)}(t')\|_2^2 \, \mathrm{d}t'$ (Line~7) and the iterative improvements of $x_2^{(k)}(t)$ during the execution of Algorithm~\ref{alg:algorithm} for the congestion game \eqref{eq:congestion_game}. As shown in Fig.~\ref{fig:congestion_fig_b}, the overshoot in $x_2^{(k)}(t)$ decreases with each iteration, leading to a reduced cost of strategy revision and an increase in the long-term payoff \eqref{eq:cost_to_go}.

Fig.~\ref{fig:RPS_states} compares the population state trajectories generated by the optimal strategy revision protocol with those obtained using the Smith ED model in the RPS game \eqref{eq:rps_game}. As shown, the population state trajectory determined by the optimal protocol exhibits reduced oscillations and converges more quickly to the Nash equilibrium of \eqref{eq:rps_game}.

Note that the congestion game \eqref{eq:congestion_game} is a contractive game satisfying the assumptions of Theorem~\ref{cor:convergence}---namely, condition~\eqref{eq:strong_contractiveness} and the existence of an interior Nash equilibrium. On the other hand, the RPS game \eqref{eq:rps_game} as a contractive game under Definition~\ref{def:contractive_games} but does not satisfy condition~\eqref{eq:strong_contractiveness}. Nevertheless, as illustrated in Fig.~\ref{fig:RPS_states}, our framework is still capable of designing the optimal revision protocol in this case, even without employing the technique discussed in Remark~\ref{remark:payoff_function_modification}. This demonstrates its broader applicability to optimal design problems in PG.

\section{Conclusions}
\label{sec:conclusions}

This paper examined the relationship between the MFG framework and the optimal strategy revision process in PG. Specifically, we established a connection between the evolution of agent state distributions in MFG, governed by the forward FK equation, and the ED model in PG. Furthermore, we demonstrated that the backward HJ equation in MFG can be utilized to define the payoff dynamics model, enabling the determination of the payoff vector for the optimal strategy revision in PG. Building on these insights, we showed how existing ED models---including Smith, replicator, and projection dynamics, along with their distributed variants---can be derived by appropriately selecting the finite time-horizon payoff function.

One limitation of the proposed approach is that agents are assumed to have full knowledge of the payoff function $F$ of the underlying population game. As a direction for future work, we plan to investigate learning-based methods in which agents, through repeated interactions within the same game, can learn the optimal strategy revision protocol without relying on this assumption. In such settings, however, the learning process may produce noisy or imperfect observations of the fitness functions. Finally, we aim to further investigate the computational tractability of the proposed algorithm for solving \eqref{eq:HJ_equation} and \eqref{eq:pairwise}, particularly as the number of strategies increases.

\bibliographystyle{unsrt}
\bibliography{references}

\appendix

\subsection{Proof of Theorem~\ref{propos:connection}} \label{proof:propos:connection}
We begin by determining the optimal state transition matrix $\alpha (t)$ within the MFG framework, where the finite time-horizon cost function in \eqref{eq:cost_mfg} is defined with $c(s(t), x(t), \alpha_{s(t)}(t)) = \sum_{j \in \mathbb{S}} \frac{q_{s(t),j}(x(t))}{2} \alpha_{s(t), j}^2 (t) - F_{s(t)}(x(t))$ and $\psi (s(T), x(T)) = -F_{s(T)}(x(T))$.\footnote{Note that, together with the continuous differentiability of $F$, the definitions of $c$ and $\psi$ satisfy Assumption~\ref{assump:mfg} provided that $q_{ij}(x(t))$ is constant or that $x(t)$ remains in $\mathrm{int}(\Delta)$ for all $t \in [t_0, T]$. Accordingly, we conduct the proof within a compact subset of $\mathrm{int}(\Delta)$, containing $x(t), \, t \in [t_0, T]$, in which $q_{ij}(x(t))$ remains strictly positive and bounded.}

To solve the optimal control problem in \eqref{eq:cost_mfg}, we employ the HJ equation \eqref{eq:mean_field system} together with the FP equation \eqref{eq:mean_field system_c}. The Hamiltonian for \eqref{eq:mean_field system} is given by:
\begin{multline*}
H_i(w(t, x(t)),x(t)) = \textstyle\min_{\alpha_i \in \mathbb{R}_{i,0}^n} \bigg( \sum_{j \in \mathbb{S}} \frac{q_{ij}(x(t))}{2} \alpha_{ij}^2  \\ 
 - F_i(x(t)) + \textstyle \sum_{j \in \mathbb{S}} \alpha_{ij} (w_j(t,x(t)) - w_i(t,x(t))) \bigg).
\end{multline*}
Under the constraint that $\alpha_{i}$ must lie in $\mathbb{R}_{i,0}^n$, the optimal state transition $\alpha_{ij} (t)$ is then given by:
\begin{align}
\label{eq:optimal_jump_intensity}
\alpha_{ij}(t) = \frac{[w_i(t,x(t)) - w_j(t,x(t))]_+}{q_{ij}(x(t))},  ~\forall i \ne j.
\end{align}  
Consequently, we derive the Hamiltonian function as follows:
\begin{multline} \label{eq:hamiltonian}
H_i(w(t, x(t)),x(t)) 
\\= - \frac{1}{2} \textstyle\sum_{j \in \mathbb{S}} \dfrac{[w_i(t,x(t)) - w_j(t,x(t))]_+^2}{q_{ij}(x(t))} - F_i(x(t)),
\end{multline}
for all $i \in \mathbb{S}$. 

By exploiting the analogy between \eqref{eq:cost_to_go} and \eqref{eq:cost_mfg}, we observe that the minimization problem in \eqref{eq:cost_mfg} corresponds to the maximization problem in \eqref{eq:cost_to_go}. Consequently, the optimal cost-to-go function $v_i(t,x(t))$ associated with \eqref{eq:cost_to_go} can be expressed as $v_i(t,x(t)) = -w_i(t,x(t))$. Substituting this relation into \eqref{eq:mean_field system} and using \eqref{eq:hamiltonian}, we obtain \eqref{eq:HJ_equation}. Furthermore, by defining the payoff vector as $p_i (t) = v_i(t, x(t))$ and using $\alpha_{ij} (t)$ derived in \eqref{eq:optimal_jump_intensity}, the optimal revision protocol is given by $\rho_{ij} (p(t), x(t)) = \frac{[p_j(t) - p_i(t)]_+}{q_{ij}(x(t))}$. Furthermore, by noting that $\alpha_{ii}(t) = -\sum_{j \in \mathbb{S}\setminus \{i\}} \alpha_{ij}(t)$, we can rewrite \eqref{eq:mean_field system_c} as:
\begin{multline*}
\dot{x}_i(t) = \textstyle\sum_{j \in \mathbb{S}\setminus \{i\}} x_{j}(t) \rho_{ji}(p(t), x(t)) \\ - \textstyle x_{i}(t) \sum_{j \in \mathbb{S}\setminus \{i\}} \rho_{ij}(p(t), x(t)),
\end{multline*}
which is equivalent to \eqref{eq:pairwise}.
This concludes the proof. \hfill $\blacksquare$

\subsection{Proof of Theorem~\ref{cor:convergence}} \label{proof:cor:convergence}

Our proof follows the main arguments presented in \cite[Theorem~3]{Diogo_2013}. According to Proposition~\ref{lemma:steadystate}, the modulo-addition stationary solution is given by $x(t) = x^\ast$ and $v(t, x^\ast) = \kappa  (t-T) \mathbf{1}_n + F(x^\ast), ~ \forall t \in [t_0, T]$, where $x^\ast$ is a Nash equilibrium of $F$.

Using \eqref{eq:strong_contractiveness} together with the property $(x(t) - x^\ast)^\top \mathbf 1_n = 0$, we can derive
\begin{align}
&\frac{\mathrm d}{\mathrm dt} (x(t) - x^\ast)^\top (v(t, x(t)) - \kappa (t-T) \mathbf 1_n - F(x^\ast)) \nonumber \\
&=\frac{\mathrm d}{\mathrm dt} (x(t) - x^\ast)^\top (v(t, x(t)) - F(x^\ast)) \nonumber \\
&\geq \epsilon \|x(t) - x^\ast\|^2 \nonumber \\
&\quad + \frac{1}{2} \!\sum_{i \in \mathbb S} (x_i(t) \!+\! x_i^\ast) \!\sum_{j \in \mathbb S}\! \frac{1}{q_{ij}(x(t))}\! \big( [v_j(t, x(t)) \!-\! v_i(t, x(t))]_+ \nonumber \\
&\qquad - [F_j(x^\ast) - F_i(x^\ast)]_+ \big)^2.
\end{align}
Furthermore, we note that
\begin{align}
&\sum_{i \in \mathbb S} (x_i(t) + x_i^\ast) \sum_{j \in \mathbb S} \frac{1}{q_{ij}(x(t))} \big( [v_j(t, x(t)) - v_i(t, x(t))]_+ \nonumber \\
&\quad - [F_j(x^\ast) - F_i(x^\ast)]_+ \big)^2 \nonumber \\
& \geq \sum_{j \in \mathbb S} (x_{i_0}(t) + x_{i_0}^\ast) \frac{1}{q_{{i_0} j}(x(t))} \big( [v_j(t, x(t)) - v_{i_0}(t, x(t))] \nonumber \\ 
&\qquad - [F_j(x^\ast) - F_{i_0}(x^\ast)]_+ \big)^2,
\end{align}
where $i_0 = \argmin_{i \in \mathbb S} v_{i}(t, x(t))$. 
Hence, it follows that
\begin{align}
&\frac{\mathrm d}{\mathrm dt} (x(t) - x^\ast)^\top (v(t, x(t)) - F(x^\ast)) \geq \epsilon \|x(t) - x^\ast\|^2 \nonumber \\
&+\frac{1}{2} \sum_{j \in \mathbb S} (x_{i_0}(t) + x_{i_0}^\ast ) \frac{1}{q_{{i_0} j}(x(t))} \big( [v_j(t, x(t)) - v_{i_0}(t, x(t))] \nonumber \\ 
&\quad - [F_j(x^\ast) - F_{i_0}(x^\ast)]_+ \big)^2. 
\end{align}

Since $x^\ast$ is an interior Nash equilibrium, we have $F_j(x^\ast) = F_{i_0}(x^\ast)$ for all $j \in \mathbb S$. Hence, using the facts that $(x(t) - x^\ast)^\top \mathbf 1_n = 0$ and $q_{ij}(x(t))$ is bounded, we obtain
\begin{align}
&\frac{\mathrm d}{\mathrm dt} (x(t) - x^\ast) ^\top (v(t, x(t)) - v_{i_0}(t, x(t)) \mathbf 1_n) \nonumber \\
&\geq \epsilon \|x(t) - x^\ast\|^2 + \sum_{j \in \mathbb S} \gamma_j [v_j(t, x(t)) - v_{i_0}(t, x(t))]^2
\end{align}
for some positive constants $\gamma_j$. Consequently, we have
\begin{align} \label{eq:key_inequaility_1}
&\| x(T/2+\tau) - x^\ast \|^2 \nonumber \\
&\quad + \| v(T/2\!+\!\tau, x(T/2\!+\!\tau)) - v_{i_0}(T/2\!+\!\tau, x(T/2\!+\!\tau)) \mathbf 1_n \|^2 \nonumber \\
&\quad +\! \| x(T/2-\tau) \!-\! x^\ast \|^2 \nonumber \\
&\quad + \| v(T/2\!-\!\tau, x(T/2\!-\!\tau)) \!-\! v_{i_0}(T/2\!-\!\tau, x(T/2\!-\!\tau)) \mathbf 1_n \|^2 \nonumber \\
&\geq 2 \int_{T/2-\tau}^{T/2+\tau} \frac{\mathrm d}{\mathrm dt} (x(t) - x^\ast) ^\top (v(t, x(t)) - v_{i_0}(t, x(t)) \mathbf 1_n) \, \mathrm dt \nonumber \\
&\geq \gamma \!\int_{T/2-\tau}^{T/2+\tau} \! \left( \|x(t) \!-\! x^\ast\|^2 \!+\! \|v(t, x(t)) \!-\! v_{i_0}(t, x(t)) \mathbf 1_n \|^2 \right) \mathrm dt
\end{align}
for some positive constant $\gamma$.

Define $F_T (\tau) = \int_{T/2-\tau}^{T/2+\tau} f_T(s) \, \mathrm d s, ~\tau \in [0, T/2-t_0]$, where $f_T(s) = \|x(s) - x^\ast\|^2 + \|v(s, x(s)) - v_{i_0} (s, x(s)) \mathbf 1_n \|^2$. From \eqref{eq:key_inequaility_1}, we obtain
\begin{align}
    \gamma F_T(\tau) \leq f_T(T/2+\tau) + f_T(T/2-\tau) = \frac{\mathrm d}{\mathrm d\tau} F_T(\tau),
\end{align}
which is equivalent to $\frac{\mathrm d}{\mathrm d\tau} \ln F_T(\tau) \geq \gamma$. Integrating this inequality over $[1, T/2-t_0]$ yields
\begin{align} \label{eq:F_inequality}
    F_T(1) \leq F_T(T/2-t_0) e^{-\gamma (T/2-t_0-1)}.
\end{align}

Let $\rho$ be the strategy revision protocol that attains the maximum of the finite time-horizon payoff $P_i(t_0, \rho)$ in \eqref{eq:cost_to_go}, so that $P_i(t, \rho) = v_i(t, x(t))$ for all $t \in [t_0, T]$ and $i \in \mathbb{S}$. Let $\rho^0$ denote the trivial revision protocol with $\rho_{ij}^0 = 0$ for all $i,j \in \mathbb S$. Since $F$ is continuous on the compact set $\Delta$, there exist nonnegative constants $L,M$ such that, for all $t \in [t_0, T]$ and $i \in \mathbb{S}$, it holds that $v_i(t, x(t)) = P_i(t, \rho) \geq P_i(t, \rho^0) \geq -L (T - t) - M$. By the same continuity (and compactness) argument, there exist nonnegative constants $L', M'$ with $v_i(t, x(t)) \leq L' (T-t) + M'$, for all $t \in [t_0, T]$ and $i \in \mathbb S$. Hence, for any $t \in [t_0, T]$ and $i \in \mathbb S$, it holds that $0 \leq v_i (t, x(t)) - v_{i_0} (t, x(t)) = v_i (t, x(t)) - \min_{i \in \mathbb S} v_{i} (t, x(t)) \leq \bar L T + \bar M$ for suitable nonnegative constants $\bar L, \bar M$. 

Consequently, using the fact that $x(t), x^\ast \in \Delta$, we obtain
\begin{align} \label{eq:key_inequaility_2}
    &\int_{t_0}^{T-t_0} \left( \|x(t) - x^\ast\|^2 + \|v(t, x(t)) - v_{i_0} (t, x(t)) \mathbf 1_n\|^2 \right) \, \mathrm dt \nonumber \\
    &\leq 2T + n \int_{0}^{T} \left( \bar L^2 T^2 + 2 \bar L \bar M T + \bar M^2 \right) \, \mathrm dt \nonumber \\
    &\leq K_3 T^3 + K_2 T^2 + K_1 T ,
\end{align}
for some nonnegative constants $K_1, K_2, K_3$. 

Using \eqref{eq:F_inequality} and \eqref{eq:key_inequaility_2}, we conclude that 
\begin{multline} \label{eq:key_inequaility_convergence}
    \lim_{T \to \infty} \int_{T/2-1}^{T/2+1} f_T(s) \, \mathrm d s \\
    \leq  \lim_{T \to \infty} (K_3 T^3 + K_2 T^2 + K_1 T) e^{-\gamma (T/2-t_0-1)} = 0.
\end{multline}

Eq. \eqref{eq:key_inequaility_convergence} implies that there exists a time instant $t(T) \in [T/2-1, T/2+1]$ such that $f_T(t(T)) \to 0$ as $T \to \infty$. By the continuous dependence of the solutions to \eqref{eq:HJ_equation} and \eqref{eq:pairwise} on their initial conditions (e.g., \cite[Theorem~3.4]{khalil_nonlinear_2002}), it follows that $\lim_{T \to \infty} f_T(T/2+1) = 0$, which, in turn, implies $\lim_{T \to \infty} \|x(T/2) - x^\ast \| = 0$. \hfill $\blacksquare$

\end{document}